\documentclass[11pt]{article}

\usepackage{a4wide}
\usepackage{amsmath, amscd, amssymb, amsthm, latexsym}
\usepackage{algorithmic}
\usepackage{algorithm}
\usepackage{graphicx}

\newcommand{\Ode}{{\mathcal O}}

\newtheorem{theorem}{Theorem}

\newtheorem{definition}[theorem]{Definition}
\newtheorem{example}[theorem]{Example}

\newtheorem{proposition}[theorem]{Proposition}

\begin{document}

\title{Efficient repeat finding via suffix arrays}

\author{
\begin{tabular}{ccc}
Ver\'onica Becher  & Alejandro Deymonnaz  & Pablo Ariel Heiber \\
\texttt{\small vbecher@dc.uba.ar} & 
\texttt{\small adeymo@dc.uba.ar} &
\texttt{\small pheiber@dc.uba.ar} \\
\end{tabular} \\
\small Departamento de Computaci\'on,
Facultad de Ciencias Exactas y Naturales,\\
\small Universidad de Buenos Aires\\
\small Argentina}

\date{April 2012}
\maketitle

\begin{abstract}
We solve the problem of finding  interspersed  maximal repeats using a suffix array construction. As it is well known, all the functionality of suffix trees can be handled
by suffix arrays, gaining practicality. Our solution improves the suffix tree based approaches for the repeat finding problem, being particularly well suited for very large inputs. 
We prove the corrrectness and complexity of the algorithms.
\end{abstract}



\section{Introduction}

Many basic problems on strings have been solved in the past using
the {\em suffix tree} data structure because of its theoretical
linear complexity bounds,~\cite{gusfield}.
However, suffix trees  proved to  be impractical
when handling very large inputs, as needed in comparative genomics
or web indexing.
The linear time and memory bounds of suffix trees
hide a large constant factor,
so many of the algorithms based on suffix trees
were superseded by algorithms based on
suffix arrays, see \cite{karkk,sadakane,puglisi}.

In this paper we focus on the classical problem of finding repeats in a given string, 
and we give a solution based on the {\em suffix array}. 
This is the kind of algorithm that everyone believes it must have already been done. 
Indeed, we published our algorithm without a formal analysis in 2009 in  \cite{biopatrones}.
The algorithm has been recently cited as the state of the art for repeat finding 
with suffix arrays \cite{vitter11}. The purpose of the present paper is to document 
the correctness of the algorithm and to give an analysis of its time and memory complexity.
In addition, the version of the algorithms we give here contain some improvements in memory usage.

We consider two variants of the repeat finding problem.
The first one is to find all {\em maximal repeats} in a given string,
where a maximal repeat is a substring that occurs at least twice,
and all its extensions occur fewer times.
The second is the problem of finding  all
the maximal substrings in a give string that occur more than once.
Using the terminology of Gusfield \cite{gusfield},
 we call them {\em supermaximal} prepeats.


\section{Notation and preliminary definitions}

\noindent
{\bf Notation}. Assume the alphabet $\+A$, a finite set of symbols.
A string is a finite sequence of symbols in $\+A$.
The length of a string $w$, denoted by $|w|$,
is the number of symbols in $w$.
We address the positions of a string $w$
by counting from $1$ to $|w|$.
The symbol in position $i$ is denoted $w[i]$,
and $w[i..j]$ represents the
substring that starts in position $i$
and ends in position $j$ of $w$,
inclusive.
A prefix of a string $w$ is an initial segment of $w$, $w[1..\ell]$.
A suffix of a string $w$ is a final segment of $w$, $w[i..|w|]$.
We say $u$ is a substring of $w$ if $u = w[i..j]$ for some $i,j$.
If $u$ is a substring of $w$ we say that $u$ occurs in $w$
at position $i$ if
$u=w[i.. i+|u|-1]$.
When $u$ is a substring of $w$, we say that $w$ is an extension of $u$.

\begin{definition}[Maximal repeat, \cite{gusfield}] \label{def:maxrep}
A {\em maximal repeat} of a string $w$ is a string that occurs more
than once in $w$, and each of its extensions occurs fewer times.
\end{definition}
\begin{example}\label{ex:ejemplo}\em
The set of maximal repeats of  $w=abcdeabcdfbcde$ is $\{ abcd,\  bcde,
bcd \}$.
Clearly $abcd$ and $bcde$ are maximal repeats occurring twice.
But also $bcd$ is a maximal repeat because it occurs three times in $w$,
and every extension of $bcd$ occurs fewer times.
There are no other maximal repeats in $w$ ($bc$, for example,
occurs three times, but since $bcd$ occurs the same number
of times, $bc$ is not a maximal repeat.)
\end{example}
\noindent
This example already shows that maximal repeats can be nested and
overlapping. A bound of the number of maximal repeats on a
string is already given in \cite{gusfield}.

\begin{theorem}[\cite{gusfield}, Theorem 7.12.1]\label{thm:cantidad}
The number of maximal repeats in a string of length $n$ is no greater
than $n$.
\end{theorem}
\noindent
This result can also be derived from Algorithm~\ref{alg:findmaxr} and Theorem~\ref{thm:output}.

\section{An algorithm to find all maximal repeats}
\label{sec:algorithm}

\newcommand{\w}[1]{w[#1..n]}
Let $w$ be a string of length $n = |w|$.
The suffix array (\cite{manber93}) of $w$ is a
permutation $r$ of the indices $1...n$ such that for
each $i < j$, $\w{r[i]}$ is lexicographically
less than $\w{r[j]}$.
Thus, a suffix array represents the lexicographic
order of all suffixes of the input $w$.
For convenience we also store the inverse permutation of
$r$ and call it $p$, namely, $p[r[i]]=i$.
As a first step of our procedure we use the fast
algorithm of~\cite{sadakane}
to build the suffix array of the input $w$
in time $\Ode(n \log n)$.

We can think each substring of~$w$ as a prefix
of a suffix of~$w$.
Suppose a maximal repeat $u$ occurs $k$ times in $w$;
then, it is a prefix of $k$ different suffixes of $w$.
Since the suffix array $r$ records the lexicographical
order of the suffixes of $w$, the maximal repeat $u$ can
be seen as a string of length $|u|$
addressed by $k$ consecutive indices of $r$.
Namely, there will be an index $i$ such that $u$
occurs at positions $r[i], r[i+1]$,..., and $r[i+k-1]$ of $w$.
The algorithm has to identify which strings addressed by consecutive
indices of the suffix array are indeed maximal repeats.

We compute the length of
the longest common prefix of each pair of consecutive
suffixes in the lexicographic order.
For this task we use the linear time algorithm of Kasai et al. \cite{lcp}.
For any position $1 \leq i < n$, $LCP[i]$ gives
the length of the longest
common prefix of $\w{r[i]}$ and $\w{r[i+1]}$.
\begin{definition}\label{def:maximal}
A substring of a given string is maximal to the right
when all its extensions to the right occur fewer times;
similarly for maximality to the left.
\end{definition}
\noindent
This definition allows for a characterization
of maximal repeats that we use in Algorithm \ref{alg:findmaxr}.

\begin{proposition}\label{prop:characterization}
A substring of a given string is a maximal repeat if and only if
it is maximal to the left and to the right.
\end{proposition}

To identify  maximal repeats we first identify the candidates
that are maximal to the right, and then we filter out
those that are not maximal to the left.
The next propositions assume the suffix array $r$ of the
input string $w$, its inverse permutation $p$, and the $LCP$ array.

\begin{proposition} \label{prop:maximalright}
A substring $u$ of $w$ is repeated and maximal to the right
if and only if there is an index $i$, $1\leq i < n$, and a number of
occurrences $k\geq 2$ such that
\\1.\  $u=w[r[i].. r[i]+|u|-1]$;
\smallskip
\\2.\ $u$ occurs exactly $k$ times in $w$,
  
-\  for each $t \in [i,i+k-2]$, $LCP[t]\geq |u|$,

-\  $LCP[i-1]<|u|$ or $i=1$,

-\  $LCP[i+k-1]<|u|$ or $i+k-1=n$,
\smallskip
\\3.\ There is $t \in [i,i+k-2]$ such that $LCP[t]=|u|$.
\end{proposition}

\begin{proof}
We first prove the implication to the right.
Assume $u$ occurs more than
once in $w$, and let $k$ be its number of occurrences.
Let $i = \min \{t \ :\  w[r[t]..r[t]+|u|-1] = u\}$
and $j = \max \{t \ : \ w[r[t]..r[t]+|u|-1] = u\}$.
By definition of $r$,
for each~$t$ such that  $i \le t \le j$,
$w[r[i]..n]$ is lexicographically before
$w[r[t]..n]$ which in turn is lexicographically before $w[r[j]..n]$.
Since the first $|u|$ characters of the suffixes addressed by
$i$ and $j$ are the same,
for every such $t$, $w[r[t]..r[t]+|u|-1] = u$.
Therefore,
the longest common prefix of any two of these strings has length at least $|u|$.
Since $j$ is the largest index and $i$ is the smallest,
for any other index $h$ outside the range of $i..j$,
either  $n-r[h] < |u|$ or $w[r[h]..r[h]+|u|-1] \neq u$,
therefore the longest common prefix
between  $w[r[i]..n]$ and any of the suffixes outside the interval
addressed by $i..j$ is necessarily less than $|u|$.
Finally, since $u$ is maximal to the right,
there must be at least two occurrences of $u$
having different extensions to the right.
This causes a value in the $LCP$ to be exactly $|u|$.

The reverse implication is implied by the following observation.
Let $u = w[r[i]..r[i]+|u|-1]$
and assume all strings $w[r[t]..r[t]+|u|-1]$ with
$i \le t < j=i+k$ share its first $|u|$ characters,
so they all  start with $u$.
Assume there is some $t$, $i\leq t < i+k-1$ such that $LCP[t] = |u|$;
therefore, either $n-max(r[t],r[t+1]) = |u|$ or $w[r[t]+|u|] \neq w[r[t+1]+|u|]$.
In either case, the repeat cannot be extended to the right.
Since the previous and next strings outside the interval addressed by $i..j$,
if they exist, have a longest common prefix of length less than $|u|$,
they do not start with $u$.
By definition of $r$, every suffix starting with $u$
has to be included in the considered interval.
Hence, every extension of $u$ to the right occurs fewer times than~$u$~itself.
\end{proof}

\begin{proposition}\label{prop:maximaleft}
A substring $u$ of $w$ maximal to the right, is also maximal to the left
if and only if there is an index $i$, $1\leq i < n$,
and a number of occurrences $k\geq 2$ such that
\smallskip
\\1. \ $u$ occurs $k$ times in $w$,
\smallskip
\\2.\ \ $w[r[i]-1] \neq w[r[i+k-1]-1]$, or

          $r[i] = 1$, or $r[i+k-1] = 1$, or

	     $p[r[i+k-1]-1]-p[r[i]-1] \neq k-1$.
\end{proposition}
\begin{proof}
We use the characterization
of maximality to the right
given in Proposition~\ref{prop:maximalright}.
Assume $u$ is maximal to the right,
let $i$ be its first index in $r$ and let $k$ be its number of occurrences.
We first show the implication to the left
by proving the contrapositive.
Suppose $u$ is not maximal to the left,
then there is a symbol $c$ such that $cu$ occurs the same number of times as $u$.
All occurrences of $cu$ also contain $u$,
hence for each of the $k$ occurrences of $u$ there is a $c$ before $u$.
Since  $u$ occurs at positions $r[i],...,r[i+k-1]$ in $w$,
 $cu$ must occur at positions $r[i]-1,...,r[i+k-1]-1$.
Thus, $r[i] \neq 1$, $r[i+k-1] \neq 1$ and
$w[r[i]] = w[r[i+k-1]] = c$.
Observe that the relative order in $r$ of
the suffixes starting with $u$ is the same as
the relative order of all the suffixes starting with $cu$,
because they are the same strings, with a $c$ added at the front.
So,  for each  $j = 0 ... k-1$, $p[r[i]-1] = p[r[i+j]-1] + j$,
and in particular, $p[r[i]-1] = p[r[i+k-1]-1] + k-1$.

For the implication to the right assume $u$ is repeated,
maximal to the right and maximal to the left.
Since $u$ occurs $k$ times, $k\geq 2$,
there is an index  $i$ such that
for every~$t$, $i\leq t\leq i+k-1$, $u=w[r[t]..r[t]+|u|-1]$.
Assume $r[i] \neq 1$, and $r[i+k-1] \neq 1$, and
$p[r[i+k-1]-1]-p[r[i]-1] = k-1$.
Since $u$ is maximal to the left, there must be at least two
positions $a$ and $b$ of $w$ witnessing that $u$ cannot
be extended to the left.
Either $w[a-1]\neq w[b-1]$ or $a=1$  or $b=1$
(recall the positions of a string are numbered starting at $1$),
while $u=w[a..a+|u|-1]=w[b..b+|u|-1]$.
But  $a=r[t]$ and $b=r[t']$  for some
indices $t, t'$ in the range $i..i+k-1$.
Since $r$ records the lexicographic ordering, if
that happens on any pair of positions $t,t'$, $i \le t < t' \le i+k-1$,
it also happens choosing
$t=i$ and $t'=i+k-1$, therefore, if none of those indices is $1$,
$w[r[i]-1] \neq w[r[i+k-1]-1]$.
\end{proof}

We call the algorithm {\bf findmaxr}. Its pseudocode
is described in Algorithm~\ref{alg:findmaxr}.
It takes as an extra parameter an integer $ml$ which is
the minimum length of a  maximal repeat to be reported (it can be set to
$1$ if desired).

\begin{algorithm}[h!]
\caption{{\bf findmaxr}($w, ml$)}
\label{alg:findmaxr}
\begin{algorithmic}
\STATE
\STATE $n := |w|$
\STATE $r :=$ suffix array of $w$
\STATE $p :=$ inverse permutation of $r$
\STATE $LCP :=$ longest common prefix array of $w$ and $r$
\STATE $S := \{u \ : \ LCP[u] < ml \} \cup \{0, n\}$
	---discard all indices of $w$ whose $LCP$ is less than $ml$---
\STATE $I :=$ permutation of $[0,n-1]$ such that $LCP[I[i]] \leq LCP[I[i+1]]$
\STATE $initial := \min \{t \ :\  LCP[I[t]] \geq ml\}$
\FOR{$t := initial$ to $n-1$ }
	\STATE  $i := I[t]$
	\STATE  $p_i:= \max \{j\ :\ j \in S \land j < i\} + 1$
	\STATE  $n_i:= \min \{j\ :\ j \in S \land j > i\}$
	\STATE  $S := S \cup \{i\}$ \\
	\IF {($p_i = 1$ or $LCP[p_i-1] \neq LCP[i]$)
		    and ($n_i = n$ or $LCP[n_i] \neq LCP[i]$) }
		\STATE ---here we have a substring maximal to the right,
			check if it is maximal to the left---
		\IF {$r[p_i] = 1$ or $r[n_i] = 1$ or $w[r[p_i]-1] \neq w[r[n_i]-1]$ or $p[r[n_i]-1] - p[r[p_i]-1] \neq n_i-p_i$ }
			\STATE ---here it is both maximal to the right and to the left---
			\STATE  report the maximal repeat $w[r[i]..r[i]+LCP[i]-1]$ and
			 its $n_i-p_i+1$ occurrences,   whose
			list of positions in $w$ are contiguous in $r$
			starting at $p_i$.
		\ENDIF
	\ENDIF
\ENDFOR
\STATE
\end{algorithmic}

\end{algorithm}
\begin{algorithm}[!h]
\caption{{\bf Maximum less than}($t$)\label{alg:max}}
\begin{algorithmic}
\STATE
\STATE ---inquire the bit tree $S$---
\REPEAT
 \IF{$t$ is a left child}
 	\STATE $t$ := rightmost node to the left of $t$ in its level
 \ELSE
 	\STATE $t$ := parent($t$)
 \ENDIF
 \UNTIL {node $t$ is set to $1$}

 \WHILE {$t$ is not a leaf}
 	\IF {right child($t$) is set to $1$ }
 		\STATE $t$ := right child($t$)
 	\ELSE
 		\STATE  $t$ := left child($t$)
 	\ENDIF
 \ENDWHILE
 \RETURN $t$
\STATE
\end{algorithmic}
\end{algorithm}

The idea of the algorithm is to treat all suffix intervals as described in
Proposition~\ref{prop:maximalright}, one by one,
but in non-decreasing order of their longest common prefix.
We use a set $S$ to keep track of
the $LCP$ values already seen at each current
step of the algorithm, and the fake indices $0$ and $n$ to treat border
cases.
The algorithm first discards
all values less than $ml$, inserting them into $S$,
hence, the required computational time in practice diminishes as $ml$
increases.

The main loop of the algorithm iterates  on the
$LCP$ values in non-decreasing order.
For each current value, the algorithm finds
the largest interval of indices in the suffix array $r$,
above and below it,
such that all the contained $LCP$ values are greater than or equal to the
current one.
Being the largest interval ensures that
all occurrences of the addressed string are taken into account.
Then the algorithm checks whether
this interval implies a string that is maximal
to the right and to the left.
Since the main loop iterates on $LCP$ values in non-decreasing
order, each current interval of indices in $r$
lays between the closest previously used indices above and below it.
Notice that the interval is only valid when the limits are strictly less
(and not equal) than the current $LCP$ value.
Each  of the intervals  considered in
 Proposition~\ref{prop:maximalright},
is addressed in the first visited of all $j$
such that $LCP[j]=|u|$.

The algorithm keeps track of the set $S$ of indices already seen;
these are the indices of $LCP$ that have already been treated
in the main loop of the algorithm.
The data structure to implement this set $S$ should be efficient for
the insertion operation in $S$, as well as the queries for the
``minimum greater than a given value'' and the
``maximum less than a given value''.
Notice that these are the most expensive operations
of the main loop, being critical for the overall time complexity.
All other operations in the main loop are $O(1)$.

We represent the set $S$ with a small data structure, which is  efficient for
the mentioned operations ensuring they take time $\Ode(\log n)$.
This data structure requires to have the number of elements of the universe
specified in advance.
In our case this is not a problem, because the universe
is the integer interval $[0,n+1]$.
We implement $S$ as a binary tree of bits with $n+2$ leaves.
Each leaf represents one of the $n+2$ elements of the universe,
the value is~$1$ if the element is in $S$, and it is $0$ otherwise.
An internal node has value~$0$ if and only if
both of its children are~$0$, otherwise
it has value~$1$.
The insertion operation only needs to update the
branch of the modified leaf,
so it can be done in $\log n$~time.

Algorithm~\ref{alg:max} gives the pseudocode to  find
the maximum less than a given~$t$.
The idea is to  go up in the tree always moving to the
rightmost node that has a chance of being set to~$1$
because of a leaf to the left of the parameter~$t$;
this rapidly increases the interval represented by the nodes,
because each move up multiplies the size of the checked interval
by at least $3/2$.
As soon as we find a node with value $1$,
we immediately go down in the tree
looking for the rightmost leaf that caused that~$1$.
Finding the minimum greater than~$t$ is analogous.

\subsection{Improved implementation saving  $n$ word\_size memory}

\begin{algorithm}[th!]
\caption{\noindent{\bf findmaxr}($w, ml$)}
\label{alg:findmaxrmem}
\begin{algorithmic}
\STATE
\STATE $n := |w|$
\STATE $r :=$ suffix array of $w$
\STATE $p :=$ inverse permutation of $r$
\STATE $LCP :=$ longest common prefix array of $w$ and $r$
\STATE $S := \{u\ :\ LCP[u] < ml \} \cup \{0, n\}$
	---discard all indices of $w$ whose $LCP$ is less than $ml$---
\STATE $I :=$ permutation of $[0,n-1]$ such that $LCP[I[i]] \leq LCP[I[i+1]]$
\STATE $initial := \min \{t\ :\ LCP[I[t]] \geq ml\}$
\STATE $lasti := -1$
\STATE $lastLCP := 0$
\STATE $curLCP := 0$
\FOR{$t := initial$ to $n-1$ }
	\STATE  $i := I[t]$
	\WHILE{$w[r[i]+curLCP] = w[r[i+1]+curLCP$}
		\STATE $curLCP := curLCP+1$
	\ENDWHILE
	\STATE  $p_i:= \max \{j \in S \land j < i\} + 1$
	\STATE  $n_i:= \min \{j \in S \land j > i\}$
	\STATE  $S := S \cup \{i\}$ \\
	\IF {($p_i = 1$ or $(lasti = p_i-1$ and $lastLCP \neq curLCP)$)}
		\STATE ---here we have a substring maximal to the right,
			check if it is maximal to the left---
		\IF {$r[p_i] = 1$ or $r[n_i] = 1$ or $w[r[p_i]-1] \neq w[r[n_i]-1]$ or $p[r[n_i]-1] - p[r[p_i]-1] \neq n_i-p_i$ }
			\STATE
		 ---here it is both maximal to the right and to the left---
			\STATE  report $w[r[i]..r[i]+curLCP-1]$ and
			 its $n_i-p_i+1$ occurrences,   whose
			list of positions in $w$ are contiguous in $r$
			starting at $p_i$.
		\ENDIF
	\ENDIF
	\STATE  $lastLCP := curLCP$
	\STATE  $lasti := i$
\ENDFOR
\STATE
\end{algorithmic}
\end{algorithm}

We give an alternative implementation of
Algorithm~\ref{alg:findmaxr}
that lowers the needed memory to
$n\ (3 \mbox{word\_size} + \log|\+A| + 2)$,
while keeping the same time complexity, cf. Theorems~\ref{thm:timefindmaxr}
and~\ref{thm:spacefindmaxr}.
In this variant we get rid of the resident
$LCP$ array before entering into the main loop, and
we efficiently deal with the three $LCP$ values needed
inside the mainloop.

Algorithm~\ref{alg:findmaxrmem} gives the pseudocode.
It starts exactly as Algorithm~\ref{alg:findmaxr},
first computing the suffix array $r$, the $LCP$ array and the set $S$.
Then it builds  the array $I$
using  the simple algorithm
{\em Inverse a permutation in place} of \cite{knuth} Algorithm~I,
Section 1.3.3, attributed there to \cite{bingchaohuang}.
This inversion requires time $\Ode(n)$
and $n$ auxiliary bits. However, the $LCP$ array is {\em not}
a permutation of $1..n$, so we first map the $LCP$
into a permutation $1..n$,
using an auxiliary array.
At this point we can free the $LCP$ space,
and the  array $I$ is constructed in the same place
of the auxiliary permutation.


%

Both,  Algorithm~\ref{alg:findmaxr}
and Algorithm~\ref{alg:findmaxrmem},
 in their main loop iterate
 the indices in non-decreasing order of the $LCP$ values.
However, Algorithm~\ref{alg:findmaxrmem}  requires that
{\em in case of tie of the $LCP$ values the indices
be iterated in increasing order}. This ordering is ensured
when we construct the array $I$.
The rest of the implementation is exactly as
Algorithm~\ref{alg:findmaxr}.


\begin{proposition}\label{prop:totalmemory}
The total memory required by Algorithm~\ref{alg:findmaxrmem}
sums up to the original input $w$,
the set $S$, and three arrays of length $n$ with
values between $0$ and $n$.
\end{proposition}
\begin{proof}
The input $w$, the set $S$, the suffix array $r$ are permanently in memory.
Besides there are at most two other integer arrays of length $n$ in memory.
First the  $LCP$  and the auxiliary permutation array.
Then $LCP$ is discarded and array $I$ is constructed in place of the
auxiliary permutation using $n$ auxiliary bits.
Then, these auxiliary bits are discarded and array $p$ is constructed.
\end{proof}

The total number of instructions involved in
the whole computation of the main loop of
Algorithm~\ref{alg:findmaxrmem}
and those in the main loop of Algorithm~\ref{alg:findmaxr}
differ in  $\Ode(1)$.
Therefore, the overall time complexity
of the two variants of {\bf findmaxr} is the same.

\begin{proposition}\label{prop:totaltimelcp}
The overall computation of
all the needed $LCP$ values in the main loop
of Algorithm~\ref{alg:findmaxrmem}
require at most $O(n)$ comparisons.
\end{proposition}
\begin{proof}
The main loop of Algorithm~\ref{alg:findmaxr}
uses three values in the  $LCP$ array.
In  Algorithm~\ref{alg:findmaxrmem} we compute the needed values,
profiting that  indices $i$ are visited in non decreasing order
of their $LCP$ values, and  those having
the same $LCP$ value are visited in increasing order
(this ordering is achieved when we construct  array $I$).
To compute $LCP[i]$ do the comparison
of the two suffixes character by character
{\em but starting from the
previously used $LCP$ value}.
Hence,  the total number of character
comparisons in the overall computation of
all the $LCP$ values is at most $2n$ (there can be at most
$n$ comparisons that yield each of the $2$ possible results).
For the comparison
$LCP[p_i-1] \neq LCP[i]$:
Since $p_i-1 = \max \{j \in S \land j < i\}$,
the index $p_i-1$ was already seen, so its $LCP$ value is
no greater than $LCP[i]$.
In case  $LCP[i]$ differs from  the  $LCP$ value of the last index seen,
 given that indexes are visited in non decreasing order of their
 $LCP$ value, and $p_i-1\in S$, hence it was already seen,  we conclude
 $LCP[p_i-1] \neq LCP[i]$.
In case $LCP[i]$ equals the $LCP$ value of the last index seen,
since indexes having the same $LCP$ value are visited in increasing order,
$LCP[p_i-1] = LCP[i]$ exactly when  $p_i-1$ coincides with the last index seen,
because $p_i-1$ is the largest seen index smaller than $i$.
For the  comparison $LCP[n_i-1] \neq LCP[i]$:
Since  indices with the same $LCP$ value in
are visited in increasing order,
and $n_i = \min \{j \in S \land j > i\}$,
namely, $n_i$ is the smallest seen index greater than $i$,
then necessarily $LCP[n_i] < LCP[i]$.
Thus, the inequality $LCP[n_i] \neq LCP[i]$ is always true.
\end{proof}

\subsection{Correctness of algorithm {\bf findmaxr}}

\begin{theorem}[Correctness]
The algorithm {\bf findmaxr} computes all occurrences of all
maximal repeats in the input string, in increasing order of length.
\end{theorem}
\begin{proof}
Consider Algorithm~\ref{alg:findmaxr}.
The main loop sequentially access the array $I$,
the permutation array for the non-decreasing order of $LCP$.
By Proposition~\ref{prop:characterization} maximal repeats are exactly
the candidates that are maximal to the right and to the left.
Propositions~\ref{prop:maximalright}
and~\ref{prop:maximaleft},
fully characterize these properties
with conditions on the data structures $LCP$, $r$, and $p$.
The algorithm checks these conditions.
\end{proof}

\subsection{Complexity of algorithm {\bf findmaxr}}

As is custom in the literature on algorithms
we express the time and space complexity
assuming integer values can be stored in a unit,
and integer additions and multiplications can be done in
$\Ode(1)$.
These assumptions make sense because the integer values
involved in the algorithm fit into the processor word size
for practical cases.
Although the algorithm is scalable for any input size,
the derived complexity bounds are guaranteed only
if the input size remains under the machine addressable size.
Otherwise, the classical logarithmic complexity charge for each integer
operation becomes mandatory.

Assume an input size of $n$ symbols.
To bound the time complexity of our main algorithm,
we first show that the set of all maximal repeats
and their occurrences can be represented in
a concise way.

\begin{theorem}\label{thm:output}
For any string $w$ of length $n$, the set of all maximal repeats
and all their occurrences
is representable in space~$\Ode(n)$.
\end{theorem}
\begin{proof}
Each iteration of the  main  loop of Algorithm~\ref{alg:findmaxr}
reports at most one maximal repeat,
 followed by the list of all its occurrences.
Each reported maximal repeat and all its occurrences
can be represented with three unsigned integers:
an index $i$ in the suffix array  $r$,
a length $\ell$,
and the number of occurrences $m$.
The reported maximal repeat is the prefix of length $\ell$
of the suffix at position $r[i]$.
Its $m$ occurrences are respectively
in positions $r[i]$,..,$r[i+m-1]$.
Each of these integers is at most $n$
(where $n$ is at most the maximum addressable memory)
and we need $n$ of them.
Assuming that these integer values
 can be stored in fixed number of bits,
 this output requires size $\Ode(n)$.
Finally, we need to store the suffix array $r$,
which contains $n$ integer values that are
a permutation of $1..n$,
so it also requires $\Ode(n)$.
The input $w$ also takes $\Ode(n)$ space, since each symbol in $\+A$
also takes $\Ode(1)$ because $|\+A| \le n$.
\end{proof}

Of course, if instead of charging a fixed number of
bits to store an integer, we count the length of
its bit representation, the total needed output space
to report all maximal repeats and occurrences in a given input
string of length $n$ becomes $\Ode(n \log n)$.
The input $w$ in this case would have the same $\Ode(n \log n)$ bound, but
probably takes a lot less because alphabet sizes are usually small compared
with $n$.

Maximum\_less\_than and Minimum\_greater\_than take $\Ode(\log n)$ time.
We prove one, the other is analogous.

\begin{proposition}\label{prop:maximumlessthan}
The time complexity of Maximum\_less\_than
a given value is~$\Ode(\log n)$.
\end{proposition}
\begin{proof}
In the repeat loop of Algorithm~\ref{alg:max}
there is at most one move to the
right for each move up in the tree,
hence, we have in total~$\Ode(\log n)$ iterations.
Then, in the while loop, every move goes down one level,
therefore there are also~$\Ode(\log n)$ total moves.
If the tree is implemented over a bit array
---in our implementation we use a bit array for each level in the tree---,
all moves are easily implemented in $\Ode(1)$;
thus, the entire running
time of each query is~$\Ode(\log n)$.
\end{proof}

\begin{theorem}[Time Complexity]\label{thm:timefindmaxr}
The algorithm {\bf  findmaxr} takes time~$\Ode(n \log n)$.
\end{theorem}
\begin{proof}
Consider Algorithm~\ref{alg:findmaxr} or  its variant as
Algorithm~\ref{alg:findmaxrmem} together with
Proposition~\ref{prop:totaltimelcp}.
All steps before the main for loop take $\Ode(n \log n)$~operations.
The main loop iterates $n$~times.
The most expensive procedure performed in the loop body
is the manipulation of
the tree for the set~$S$,
which requires at most  $\Ode(\log n)$~operations, cf.
Proposition~\ref{prop:maximumlessthan}.
\end{proof}

\begin{theorem}[Space complexity]\label{thm:spacefindmaxr}
For an input of size $n$,
algorithm {\bf findmaxr} uses $\Ode(n)$ memory space.
More precisely, it uses
$n\ (3 \mbox{word\_size} + \log |\+A| + 2) + \Ode(1)$ bits.
\end{theorem}
\begin{proof}
Consider Algorithm~\ref{alg:findmaxrmem}.
The whole input $w$ is allocated in memory.
Since it contains $n$ symbols,
its memory usage is $n \log |\+A|$ bits.
By Proposition~\ref{prop:totalmemory}
the total amount of memory needed
is for the input $w$, the set $S$ and three
integer arrays of length~$n$ with values
between $0$ and~$n$.
The described tree for the set~$S$ has $2n+1$~nodes,
implemented with an array of $2n+1$~bits.
Then, the exact memory space of all the data structures
is $n\ (3\ \mbox{word\_size} + \log |\+A| + 2)$ bits.
The local variables are counted in the $\Ode(1)$ term.
\end{proof}


\section{An algorithm to find all supermaximal repeats}

\begin{algorithm}[th]
\caption{{\bf findsmaxr}($w, ml$)}
\label{alg:findsmaxr}
\begin{algorithmic}
\STATE
\STATE $n := |w|$
\STATE $r :=$ suffix array of $w$
\STATE $LCP :=$ longest common prefix array of $w$ and $r$
\STATE $up := 1$
\FOR{$i := 2$ to $n-1$}
	\IF{$LCP[i] > LCP[i-1]$}
		\STATE $up := i$ ---starting position in $LCP$ for the set of
			local maximum values---
	\ELSE
		\IF{$LCP[i] \neq LCP[i-1] \land LCP[i-1] \geq ml$}
         \STATE ---the indices from $up$ to $i-1$
         give a local maximum in $LCP$ of appropriate length---
		\IF{$\#\{w[r[j]-1]\ :\ up \le j \le i-1 \land r[j] > 1\} = i-up+1$}
				\STATE ---check that all previous characters are different---
				\STATE report $i-up+1$ supermaximal repeats
				of size $LCP[up]$ whose
				\STATE list of positions in $w$
				are contiguous in $r$ starting
				at $up$.
			\ENDIF
		\ENDIF
	\ENDIF
\ENDFOR
\STATE
\end{algorithmic}
\end{algorithm}

We solve here the problem of finding the maximal substrings
of the input that are repeated.
These are a subset of the maximal repeats
(in the sense of Definition~\ref{def:maxrep})
in the input string such that
none of their extensions is also a maximal repeat,
called {\em supermaximal repeats}.
For instance, in Example~\ref{ex:ejemplo}
the set of all maximal repeats
in string $w=abcdeabcdfbcde$
is $\{abcd, bcde, bcd\}$.
While the set of supermaximal repeats
is just $\{abcd , bcde \}$
since $bcd$ is a substring of $abcd$
(and also of $bcde$).

\begin{definition}[Supermaximal repeats, \cite{gusfield}]\label{def:smaxrep}
A string $u$ is a {\em supermaximal repeat} in $w$
if $u$ is a substring that occurs at least twice in $w$ and each extension of
$u$ occurs at most once in~$w$.
\end{definition}

Algorithm~\ref{alg:findsmaxr}, called {\bf findsmaxr},
finds all supermaximal repeats in a given string $w$.
It is similar to the algorithm \textbf{findmaxr}
of the previous section, but simpler.
We define an analogous notion of maximality to the right and to the left
for this case.

\begin{definition}\label{def:smaximal}
A substring $u$ that occurs more than once in $w$
is supermaximal to the right if all of
its extensions to the right occurs at most once in $w$.
Supermaximality to the left is defined analogously.
\end{definition}

\begin{proposition}\label{prop:local}
A substring of $w$ is supermaximal to the right
if and only if it is addressed
by consecutive indices in the suffix  array
denoting a local maximum in $LCP$.
\end{proposition}
\begin{proof}
A substring $u$ of $w$ is supermaximal to the right,
cf. Definition~\ref{def:smaximal}, exactly when
there is a set of at least two consecutive suffixes addressed by $r$
such that all the addressed suffixes have the same maximal common prefix,
which is longer than that of any two other two suffixes,
one taken in the set and one outside the set (note that any pair outside
the set will have a common prefix of the same size that is necessarily
different to the one we are considering).
So, there is a minimum $i$, $1\leq i<n$,
and there is a number of occurrences $k\geq 2$ such that
\\-\ $u=w[r[i]..r[i]+|u|-1]$,
\\-\ $LCP[i] = LCP[i+1] = ... = LCP[i+k-2]$,
\\-\ ($LCP[i] > LCP[i-1]$ or $i = 1$ ) and
       ($LCP[i] > LCP[i+k-1]$
 or $i+k-1$ = $n-1$).
\end{proof}

Also, if a substring of $w$ is not supermaximal to the left,
there are at least two extensions by one symbol to the left
that are equal.

\begin{proposition} \label{prop:countleft}
A substring $u$ in $w$ is supermaximal to the left
if and only if it occurs at least twice but at most $|\+A|$ times in $w$ and
and all its extensions by one symbol to the
left are pairwise different.
\end{proposition}
\begin{proof}
Supermaximality to the left
requires that the previous symbol of each of
the addressed suffixes be pairwise different.
If any two of these suffixes had the same previous symbol,
then there would be an extension to the left
that is repeated twice.
\end{proof}

\begin{proposition}\label{prop:maximalcharacterization}
A substring of a given string is a supermaximal repeat if and only if
it is supermaximal to the right and to the left.
\end{proposition}

The algorithm {\bf findsmaxr} first identifies the candidates
that are supermaximal to the right
by finding the set  of consecutive positions in $r$
that have local maximums in $LCP$.
Then it filters out the candidates
that are not supermaximal to the left.
To check supermaximality to the left,
we construct the set of previous symbols (symbols that occur immediately before
each suffix on the set).
Since the universe of this set is the size of the alphabet $\+A$,
it can easily be efficiently implemented in a boolean array
of the size of $\+A$.
We take as an extra parameter an integer $ml$ which is
the minimum length of a supermaximal repeat to be reported
(it can be set to $1$ if desired).
The pseudocode is given in Algorithm~\ref{alg:findsmaxr}.

\subsection{Correctness of {\bf findsmaxr}}

\begin{theorem}[Correctness]
The algorithm {\bf findsmaxr} computes all supermaximal repeats in the input
string.
\end{theorem}
\begin{proof}
Consider Algorithm~\ref{alg:findsmaxr}.
The main loop accesses sequentially  the array $LCP$.
By Proposition~\ref{prop:maximalcharacterization} supermaximal repeats are exactly
the candidates that are supermaximal to the right and to the left.
Propositions~\ref{prop:local}
and~\ref{prop:countleft},
characterize these properties
with conditions on the data structures $LCP$ and $r$.
The algorithm checks exactly these conditions.
\end{proof}

\subsection{Complexity of {\bf findsmaxr}}

\begin{theorem}[Time complexity]
Given the suffix array for an input string
of length $n$, {\bf findsmaxr}
computes all supermaximal repeats in time
$\Ode(n)$.
\end{theorem}
\begin{proof}
The main loop iterates $n-2$ times.
Inside the main loop, all steps are clearly done in $\Ode(1)$ except the
construction of the set on the innermost if clause.
This is done by inserting each element into the set, represented as a boolean
array, and maintaining the size appropriately. Each insertion takes $\Ode(1)$.
Thus, the total number of insertions over all constructions of this set
is the sum of the number of suffixes in each local maximum.
Since no suffix can be in two different local maximums,
this total is less than $n$.
\end{proof}

\begin{theorem}[Space complexity]
For an input of size $n$,
algorithm {\bf findsmaxr} uses $\Ode(n)$ memory space.
More precisely, it uses
$n\ (2\ \mbox{word\_size} + \log |\+A| + 2)  + |\+A| + \Ode(1)$ bits.
\end{theorem}
\begin{proof}
Consider Algorithm~\ref{alg:findsmaxr}.
The whole input  $w$ is allocated in memory.
Since it contains $n$ symbols,
its memory usage is $n \log |\+A|$ bits.
The data structures $r$ and $LCP$ are
 arrays of length~$n$ whose elements are
between $0$ and~$n$, and fit within the processor word.
The described array for the set~$A$  has $|\+A|$~bits.
As before, since $|\+A| \le n$, the term $\log |\+A|$ is considered
$\Ode(1)$, as is the term $\mbox{word\_size}$.
\end{proof}


\section{Implementation and experimental results}
\label{sec:output}

\begin{table}[!ht]
\caption{Input data set used for comparison.}
\begin{center}
\label{table:inputs}
\begin{tabular}{llrrrr}
\hline
File & Description & Size (bytes) \\
\hline
linux & The Linux Kernel 2.6.31 tar file                     & 365 711 360 \\
HS-ch1 & Homo-sapiens chromosome 1, from NCBI 36.49          & 251 370 600 \\
ecoli & The file E.coli of the large Canterbury corpus       & 4 638 690 \\
bible & The file bible.txt of the large Canterbury copus     & 4 047 392 \\
world & The file world192.txt of the large Canterbury corpus & 2 473 400 \\
a2M & The letter ‘a’ repeated two million times.             & 2 000 000 \\
\hline
\end{tabular}
\end{center}
Canterbury corpus can be found at
\texttt{\small http://corpus.canterbury.ac.nz/descriptions/{\#}cantrbry},
while the FASTA files for the NCBI 36.49 DNA human genome are downloadable at
\texttt{\small ftp://ftp.ensembl.org/pub/release-49/fasta/homo\_sapiens/dna/}
\end{table}

\begin{table}[!ht]
\caption{Input size in bytes and running time in seconds of the three processes}%
\begin{center}
\label{table:runningtime}
\begin{tabular}{lrrrr}
\hline
File &  Size & SA & findmaxr & findsmaxr \\
\hline
linux &  365 711 360 & 671.785 & 275.834 & 66.587 \\
HS-ch1 & 251 370 600 & 798.583 & 197.545 & 100.714 \\
ecoli &  4 638 690 & 4.586 & 2.543 & 1.636 \\
bible &  4 047 392 & 3.201 & 2.085 & 0.750 \\
world &  2 473 400 & 1.670 & 1.224 & 0.377 \\
a2M &    2 000 000 & 2.423 & 29.135 & 0.111 \\
\hline
\end{tabular}
\end{center}
\end{table}

We implemented  {\bf findmaxr} and  {\bf findsmaxr}  in C (ANSI C99), for a
32 or 64 bits machine. 

For {\bf findmaxr}  the memory space requirement is
$n\ (3\ \mbox{word\_size} + \log |\+A| + 2)$
bits for an input of $n$.
So, for $\+A$ the ASCII code and storing indices in 32 bits variables
this becomes a total memory requirement of $13.25 n$ bytes.
In a 64 bits processor and 8 Gb RAM installed,
the tool runs inputs of size up to
$\sim 618$ Mb, without any swapping.
Somewhat larger inputs can also be run efficiently because
some swapping does not affect the running time.  
Notice that the in main loop of the Algorithm~\ref{alg:findmaxrmem},
the array~$I$ is only used one element at a time,
so it can be easily handled by swap memory.

For {\bf findsmaxr}  the memory space requirement is
$n\ (2\ \mbox{word\_size}  + \log |\+A| + 2) + |\+A| $.
Again, for $\+A$ the ASCII code and indices in 32 bits,
the memory requirement is $9.25 n$ bytes, which allows inputs of
$\sim 885 Mb$ in the same configuration.

We tested this implementation on large inputs, using an
Intel\textregistered\ Core\texttrademark 2 Duo E6300 (only one core),
running at 1.86GHz with 8GB RAM (DDR2-800) under Ubuntu linux for 64 bits.

The programs were compiled with the GCC compiler version 4.2.4,
with option -O2 for normal optimization.
The reported times are user times, counting \emph{only} the time consumed
by the algorithm, not including the time to load the input from disk.
The input files used are described in Table~\ref{table:inputs}.
The files are chosen to demonstrate the behavior of the program
for different kinds of natural data as well as degenerated cases.
We used input sizes of the order of magnitude of the calculated limit,
validating the performance in those cases. See \cite{biopatrones} for
experiments that push this to the actual limit.

Table~\ref{table:runningtime} reports, for each case, the input size expressed
in bytes and the running time expressed in seconds of three processes:
{\em SA} our implementation of the  suffix array construction of
	\cite{sadakane},
{\em findmaxr}  the computation of all maximal repeats with $ml = 1$, and
{\em findsmaxr} the computation of all supermaximal repeats with $ml = 1$.
The time of both algorithms does not include the time for the
construction of the
suffix array. To know the time to compute the repeats from scratch,
simply add both times.

\subsection{Some remarks}

The first step in our algorithms is the computation of the suffix array
of the input string. We chose the smart algorithm of \cite{sadakane}
that, profiting from the fact that the sorting is done on suffixes
and not on arbitrary strings,
it achieves a fast $\Ode(n \log n)$ time and
requires $\Ode(n \log n)$ memory if the bit representation
of each integer is accounted for.

An information-theoretic argument tells that $\Omega(n \log |\+A|)$
bits are required to represent the
permutation given by the suffix array,
because there are $|\+A|^n$ different
strings of length $n$ over the alphabet $\+A$;
hence, there are at most that many different suffix arrays,
and we need
$\Omega(n \log |\+A|)$ bits to distinguish between them.
With suffix arrays $n \log n$ bits are used for this instead.
{\em Compressed suffix arrays} were used to get closer to the lower
bound (\cite{grossi2000,compsadakane,sadakane2002}), 
but at the cost of increasing the computational time 
\cite{deymonnaz} (for instance, 
when mapping substrings with many ocurrences).

The algorithm of Lippert \cite{lippert} is based on a compressed suffix array
and finds all repeats of {\em a given length} in the input string.
In contrast to our algorithm, Lippert's
does not indicate whether a repeat is maximal.
His experiments  show the thta his solution requires much more time than ours.

The problem of finding repeats in large inputs occurs  in genomic
sequence analysis.
A list of the most  popular repeat finders for
genomic sequences appears in the  survey  \cite{sahal}.
Leaving aside heuristic and
library based methods such as RepeatMasker (2009)\nocite{repeatmasker},
existing methods based on the suffix array still use  the length
of repeats as a parameter at each pass.
The algorithm of  \cite{insights}
constructs  a suffix array after building first a suffix tree of the input,
hence yielding a very poor performance in terms of time and memory.
A popular tool is  {\small REP}uter,  \cite{reputer,reputermanifold}.
It allows a very limited input size
and its memory requirement depends on the repeat length
and the number of occurrences in the input.
In the worst case inputs are limited to RAM size/45.
Since the output given by {\small REP}uter is not factorized,
it becomes very large, needing $\Ode(n^2)$ space for inputs of size~$n$.

\bibliographystyle{plain}
\bibliography{patrones}

\end{document}